\documentclass[11pt]{article}
\usepackage{fullpage}
\usepackage{amsmath,amsthm,amsfonts,amssymb, clrscode}
\usepackage{mathptmx}
\usepackage{mathrsfs}

 \newcommand{\ord}{\textrm{ord}}
\newcommand{\abs}[1]{\vert #1 \vert}
\newcommand{\gen}[1]{\langle #1 \rangle}
\newcommand{\ket}[1]{\vert #1 \rangle}

\newcommand{\Comp}{\mathbb C}
\newcommand{\Int}{\mathbb Z}
\newcommand{\Irr}[2]{\textrm{Irr}(#1,#2)} 
\newcommand{\IrrChar}[1]{\textrm{Char}(#1)} 
\newcommand{\Subgroups}[1]{\mathscr{H}_{#1}} 
\newcommand{\normal}[1]{\overline{#1}} 
\newcommand{\MEMB}{\textrm{MEMB}} 

\newcommand{\ceil}[1]{\lceil #1 \rceil}

\newcommand{\st}{\:\vert\:}

\newtheorem{theorem}{Theorem}
\newtheorem{proposition}{Proposition}
\newtheorem{definition}{Definition}

\newtheorem{corollary}{Corollary}

\newcommand{\C}{\mathbb{C}}
\newcommand{\Z}{\mathbb{Z}}
\newcommand{\F}{\mathbb{F}} 
\newcommand{\Q}{\mathbb{Q}}
\newcommand{\poly}{\text{poly}}
\renewcommand{\vec}[1]{\mathbf{#1}}

\newcommand{\tr}{\text{tr}}
\DeclareMathOperator{\Exp}{\mathbb{E}}
\newcommand{\GL}{\textsf{GL}}

\begin{document}

\author{%
Scott Aaronson \\ \texttt{aaronson@csail.mit.edu} \\ Massachusetts Institute of Technology%
\and Fran{\c c}ois Le Gall \\ \texttt{legall@qci.jst.go.jp} \\
		Japan Science and Technology Agency%
\and Alexander Russell \\ \texttt{acr@cse.uconn.edu} \\University of Connecticut\\%
\hspace{65mm}
\and Seiichiro Tani \\ \texttt{tani@theory.brl.ntt.co.jp} \\
		Japan Science and Technology Agency \\
                    and
                    NTT Corporation
}
\title{The One-Way Communication Complexity\\ of Group Membership}

\maketitle

\begin{abstract}
This paper studies the one-way communication complexity of the \emph{subgroup membership problem}, a 
classical problem closely related to basic questions in quantum computing.  Here
Alice receives, as input, a subgroup $H$ of a finite group $G$; Bob receives an element $x \in G$.  Alice is permitted to send a single message to Bob, after which he must decide if his input $x$ is an element of $H$.
We prove the following upper bounds on the classical communication complexity
of this problem
in the bounded-error setting:
\begin{enumerate}
\item The problem can be solved with $O(\log |G|)$ communication, provided the subgroup $H$ is normal.
\item The problem can be solved with $O(d_{\max} \cdot \log |G| )$ communication, where $d_{\max}$ is the maximum of the dimensions of the irreducible complex representations of $G$.
\item For any prime $p$ not dividing $|G|$, the problem can be solved with $O(d_{\max} \cdot \log p )$ communication, where $d_{\max}$ is the maximum of the dimensions of the irreducible $\F_p$-representations of $G$.
\end{enumerate}
\end{abstract}

\section{Introduction}
\paragraph{Background}
The power of
quantum computing in various settings has been gradually clarified by
many researchers: some problems can be solved on quantum computers
much more efficiently than on classical computers, while others
cannot. \ One computational model that has been extensively studied in this light
is the communication complexity model. \
In particular, {\em one-way communication} is one of the  simplest settings
but it has rich connections to areas such as information theory, coding theory, on-line computing, and learning theory.
\ Therefore,
its quantum version has  then been  the  target of intensive research
\cite{Aaronson05, Iwama+ICALP07,KlauckSICOMP07, Gavinsky+STOC07}.

Let $f:X\times Y\to \{0,1\}$ be a Boolean function, where $X$ and $Y$
are arbitrary sets.  The one-way communication task associated to
$f$ is the following: Alice has an input $x\in X$, Bob has an input
$y\in Y$ and the goal is for Bob to output $f(x,y)$. The
assumption here is that only one message can be sent, from Alice to
Bob, and the communication cost of a protocol is the number of
bits of this message on the worst-case input.  We say that a protocol for $f$ has completeness error $\varepsilon$
if it outputs $1$ with probability at
least $1-\varepsilon$ whenever $f(x,y)=1$, and soundness error $\delta$ if it outputs $0$ with probability at
least $1-\delta$ whenever $f(x,y)=0$.  The one-way classical
bounded-error communication complexity of $f$, denoted by $R^1(f)$, is
the minimum, over all protocols $P$ for $f$ with completeness and soundness error $1/3$, of the
communication cost of $P$.  The one-way quantum bounded-error
communication complexity of $f$, denoted by $Q^1(f)$, is defined
similarly, but a quantum message
can be used
this time from Alice to Bob,
and the number of qubits of
the message
is considered (in this
paper we suppose that there is no prior entanglement and no shared
randomness between Alice and Bob).  Obviously for any function $f$,
the relation $Q^1(f)\le R^1(f)\le \ceil{\log_2\abs{X}}$ holds.

One of the main open problems in quantum communication complexity
is to understand how large the gap between $R^1$ and
$Q^1$ can be.  For partial functions (functions restricted to some
domain $R\subset X\times Y$ or, equivalently, functions with a promise
on their inputs), an exponential separation between
these two quantities has been shown recently
in~\cite{Gavinsky+STOC07}. \ However the situation for total functions is
far less clear: the largest gap known is an asymptotic factor of $2$
\cite{Winter04}.

In the exact setting, i.e.,
the setting where no error and no giving up are allowed,
the quantum and classical one-way communication
complexities are known to be the same for any total function~\cite{KlauckSICOMP07}. \ In the unbounded-error setting,
i.e., the setting where any error probability less than 1/2 is allowed,
it is known
that the gap is exactly a factor $2$ for both partial and total
functions~\cite{Iwama+ICALP07}. \
Although bounded-error is a notion between the exact and unbounded-error,
we stress that the bounded-error setting usually behaves quite differently from the other
two in the case of total functions, e.g., for two-way
communication there is a quadratic gap in the bounded-error setting~\cite{Kalyanasundaram+92SIDM,Aaronson+05} whereas in the exact setting no gap is known and,
in the unbounded-error setting, the gap is again exactly a factor 2~\cite{Iwama+ICALP07}.

Note also that for total functions in the bounded-error setting, quadratic gaps are known in the two-way model~\cite{Kalyanasundaram+92SIDM,Aaronson+05} and exponential gaps are known in the simultaneous message-passing model~\cite{Newman+96STOC,Buhrman+PRL01}; and these models are respectively stronger and weaker than the one-way model. \ Thus, whether a superlinear gap between $R^1$ and $Q^1$ can be achieved for some total function is an intriguing question.

\paragraph{The subgroup membership function}
Many of the problems for which quantum computation is more powerful than classical computation have group-theoretic structure. \
In particular, Watrous \cite{WatrousFOCS00} has used the subgroup membership problem (as a computational problem) to separate the complexity
classes $MA$ and $QMA$
relative to an oracle. \
Inspired by Watrous's work \cite{WatrousFOCS00}, we propose the subgroup membership function as a candidate to show a superlinear 
gap between $R^1$ and $Q^1$. \
Let $G$ be any finite group, and let $\Subgroups{G}$ denote the set of subgroups of
$G$. \ Then the subgroup membership function for $G$, denoted by $\MEMB_G$, is the function with
domain
$\Subgroups{G}\times G$  such that
$$\MEMB_G(H,y)=\left\{
\begin{array}{ll}
1 &\textrm{ if }y\in H\\
0 &\textrm{ if }y\notin H.
\end{array}
\right.$$
For any group $G$, the upper bound $\abs{\Subgroups{G}}\le 2^{(\log_2\abs{G})^2}$ follows
easily from the fact that any subgroup of $G$ is generated by at most $\log_2\abs{G}$ elements.\footnote{Borovik, Pyber, and Shalev~\cite{Borovik:1996it} have improved this naive bound to $|G|^{(1/4 + o(1)) \log_2 |G|}$.} \
Furthermore, there exist families of groups $G$ such that $\abs{\Subgroups{G}}=2^{\Omega((\log\abs{G})^2)}$: for example, the
abelian groups $G=\Int_2^r$ with $r \ge 1$. Thus there exist groups $G$ for which
the ``trivial protocol,'' wherein Alice simply sends Bob the name of her
subgroup, requires $\Theta((\log \abs{G})^2)$ communication. \
The one-way classical communication complexity of the function $\MEMB_G$ was previously considered by Miltersen et al.~\cite{Miltersen+JCSS98},
who showed that for the family of groups $G=\Int_2^r$, any one-way protocol with perfect soundness and completeness error $1/2$ requires $\Omega((\log\abs{G})^2)$-bit communication. \ For certain groups $G$, we conjecture that $\Omega((\log\abs{G})^2)$-bit communication is needed even if the completeness and soundness errors are both $1/3$.

On the other hand, there is a simple {\em quantum} one-way protocol, using $O(\log\abs{G})$-bit communication, by which Bob can compute
 $\MEMB_G$ with perfect completeness and constant soundness for any group $G$. \ In this protocol---inspired by \cite{WatrousFOCS00}---Alice sends the quantum state $\ket{H}=\abs{H}^{-1/2} \sum_{h\in H}\ket{h}$.
Bob then creates the state $\frac{1}{\sqrt 2}(\ket{H}\ket{0}+\ket{yH}\ket{1})$
where $\ket{yH}=\abs{H}^{-1/2}\sum_{h\in H}\ket{yh}$, applies a Hadamard gate on the
last register, and measures it in the basis $\{\ket{0},\ket{1}\}$ to decide which of $\ket{H}=\ket{yH}$ and $\langle H\ket{yH}=0$ holds.

Thus, proving that there exists a family of groups $G$ such that
$R^1(\MEMB_G)=\Omega((\log\abs{G})^2)$ would lead to a quadratic separation between
$R^1$ and $Q^1$ for a total function. \
In other words,  a major objective has been to prove a 2-sided-error version of the lower bound by Miltersen et al.~\cite{Miltersen+JCSS98}
mentioned above. \
Apart from the goal of proving a separation between $R^1$ and $Q^1$, we believe that the communication complexity
of deciding subgroup membership is interesting in itself, since the latter is a key task in most group-theoretic
computational problems.

\paragraph{Overview of our results}
In this paper we present three upper bounds on the one-way classical communication complexity of the subgroup membership function:
\begin{itemize}
\item
We give a classical protocol using
$\ceil{\log_2\abs{G}}$-bit communication, with perfect completeness and constant soundness, for the subgroup membership function
in the case where Alice's subgroup $H$ is normal. \ This suggests that in order to obtain a separation between $R^1$
and $Q^1$ using the subgroup membership problem, one must consider groups with many
nonnormal subgroups. 
We also present a lower bound which is tight for some families of groups.
Notice that this situation appears to be similar to the status of the Hidden Subgroup
Problem: there exists an efficient quantum algorithm solving the problem 
in the case where the
hidden subgroup is normal~\cite{Hallgren+SICOMP03}; without the
normality condition, however, very little is known.  Our results rely on the
theory of characters of finite groups and especially on the connection
between kernels of irreducible characters and normal subgroups, as did the
algorithms of~\cite{Hallgren+SICOMP03}.

\item Let $p$ be a prime not dividing $\abs{G}$. Then we show that
$R^1(\MEMB_G)=O(d_{\max}^p \cdot \log p )$, where $d_{\max}^p$ is the
maximum dimension of an irreducible $\F_p$-representation of $G$.
This result uses a beautiful
characterization of the subspaces of the group algebra $\F_p[G]$
stabilized by $H$.
We remark that for any group $G$ of exponent $m$ (which is to say that $g^m = 1$ for all $g \in G$), we have $d_{\max}^p \leq d_{\max}^0 \ord_m(p)$, where $d^0_{\max}$ is the maximum dimension of a complex irreducible representation of $G$ and $\ord_m(p)$ is the order of $p$ in $\Z_m^*$, the multiplicative group of the integers relatively prime to $m$. In particular, as there is always a prime $p$ of size $O(\log |G|)$ relatively prime to $|G|$, this protocol has complexity no more than $O(d_{\max}^0 \cdot m \cdot \log\log |G|)$.

\item
Finally, we show that $R^1(\MEMB_G)=O(d_{\max}^0 \cdot \log |G| )$, where
$d_{\max}^0$ is the maximum dimension of an irreducible complex representation of $G$. \
This upper bound is obtained by a protocol that mirrors the technique utilized in the modular case by suitably discretizing the vector space $\Comp^d$ and controlling ``geometric expansion'' around invariant spaces. \ One corollary is that any family of groups with an abelian subgroup of constant index has a protocol with complexity $O(\log |G|)$. In particular, for groups such as $G=\Z_2 \ltimes \Z_2^{n}$, the action of $\Z_2$ on $\Z_2^n$ being to reverse the order of the coordinates, we have $R^1(\MEMB_G)=O(\log |G|)$.

\end{itemize}

These results suggest a nontrivial connection between the representation theory of the group $G$ and 
the subgroup membership problem, and provide natural candidates for which a superlinear separation between $R^1(\MEMB_G)$ and $Q^1(\MEMB_G)$ 
may be obtained: groups with large irreducible representations and many nonnormal subgroups, e.g.,~the symmetric group.
\vspace{10pt}

\section{Preliminaries}
We assume the reader is familiar with basic
concepts of group theory. \ Here we introduce some notions from representation theory that we will need.
\ In this paper, $G$ always denotes a finite group and $1$ denotes its identity element.

\paragraph{Group representations}

Let $\F$ be a field whose characteristic does not divide the order of $G$ (so the characteristic of $\F$ could be zero).
An $\F$-\emph{representation} $\rho$ of $G$
is a homomorphism from $G$ to $\GL(V)$, the group of invertible linear transformations over a vector space $V$ (over the field $\F$).
The \emph{dimension} of $\rho$ is the dimension of $V$. We say that a representation $\rho: G \rightarrow \GL(V_\rho)$ is \emph{irreducible} if the only subspaces of $V_\rho$ simultaneously fixed by the entire family of linear operators $\rho(g)$ are the trivial ones: $\{\vec{0}\}$ and $V_\rho$.

The \emph{group algebra} $\F[G]$ is the $\F$-algebra of formal sums
$$
\sum_{g \in G} \alpha_g \cdot e_g\,,\quad \alpha_g \in \F\,,
$$
with coordinatewise addition and multiplication defined by linearly extending the rule $e_{g} \cdot e_{h} = e_{gh}$.
Note that $\F[G]$ has dimension $|G|$ as a vector space over $\F$.
The natural action of $G$ on the group algebra defines the \emph{regular representation}: the action of $x\in G$ on a vector
$\vec{v}=\sum_{g \in G} \alpha_g \cdot e_{g}$ in $\F[G]$ is denoted by $x\vec{v}$ and defined as
$$x\vec{v}=\sum_{g \in G} \alpha_g \cdot e_{xg}\,.$$

Now, if $H$ is a subgroup of $G$, let
$$
I_H = \{ \vec{v} \in \F[G] \mid \text{$h\vec{v} = \vec{v}$ for all $h \in H$} \}
$$
be the subspace of $H$-invariant vectors of $\F[G]$.
It is easy to check that a vector $\vec{v}$ lies in $I_H$ if and only if $\vec{v}$ is constant on each left coset of $H$ in $G$.
Let $\mathcal S_{G:H}$ be the set of right cosets of $H$ in $G$.
The vectors $\vec{v_S}=\sum_{g\in S} e_g$ for $S\in\mathcal S_{G:H}$ form a basis of $I_H$ and thus
$$
\dim I_H = [G:H]\,,
$$
where $[G:H]=\abs{\mathcal S_{G:H}}=\abs{G}/\abs{H}$ denotes the index of $H$ in $G$.

A theorem of Maschke's~(see, e.g.,~\cite{CR06,Serre77}) asserts that $\F[G]$ is semi-simple,
i.e.,~$\F[G]$ can be written as the direct sum of a family of irreducible representations. In this case, a theorem of Wedderburn's~\cite{Serre77,CR06} asserts that each irreducible representation appears with multiplicity equal to its dimension:
$$
\F[G]=\bigoplus_{\rho\in \Irr{G}{\F}}V_\rho^{\oplus d_\rho}\,,
$$
where $\Irr{G}{\F}$ denotes the set of (representatives of) all the irreducible $\F$-representations $\rho\colon G\to \GL(V_\rho)$ and $d_\rho$ denotes the dimension of $\rho$. 
If $I_H(\rho)$ is the subspace of $V_\rho$ pointwise fixed by $H$, we see that
$$
I_H =\bigoplus_{\rho\in \Irr{G}{\F}}[I_H(\rho)]^{\oplus d_\rho}\,,
$$
and conclude that
\begin{equation}
\label{eq:stable}
\sum_{\rho\in \Irr{G}{\F}} d_\rho \dim I_H(\rho) = [G:H]\,.
\end{equation}

\paragraph{Complex characters}

Let $\F$ be the complex field $\Comp$.
For any $\Comp$-representation $\rho$ of $G$, the character of $\rho$ is the function $\chi\colon G\to\Comp$ such that
$\chi(g)=\tr(\rho(g))$ for any $g\in G$, where $\tr$ denotes the trace. Characters are conjugacy class functions: the relation $\chi(gg'g^{-1})=\chi(g')$ holds for any two elements $g,g'$ of $G$.
Moreover,  the value $\chi(1)$ is the dimension of the representation $\rho$.
The kernel of $\chi$, denoted by $\ker(\chi)$, is defined as $\ker(\chi)=\{g\in G\st \chi(g)=\chi(1)\}$. It is easy to see that $\ker(\chi)$ is a subgroup of $G$.

A character is said to be irreducible if it is the character of an irreducible representation.
Denote by $\IrrChar{G}$ the set of irreducible (complex) characters of $G$.
The relation $\sum_{\chi\in \IrrChar{G}}[\chi(1)]^2=\abs{G}$ is well-known and
implies the inequality $\abs{\IrrChar{G}}\le \abs{G}$.
Let $H$ be a normal subgroup of $G$. Denote
$$\Lambda_H=\{\chi\in \IrrChar{G}\st H\le \ker(\chi)\}.$$
Then the relation
\begin{equation}\label{equation:char}
\sum_{\chi\in \Lambda_H}[\chi(1)]^2=[G:H]
\end{equation}
holds (see, e.g.,~\cite{Isaacs76}).

\section{Normal subgroups}
In this section we give an efficient classical protocol computing the subgroup membership function
in the special case where Alice's subgroup $H$ is normal. Our protocol is actually more general: 
we show that one can decide efficiently membership in the normal closure of $H$, denoted by $\normal{H}$ 
(the smallest normal subgroup of $G$ containing $H$).

The protocol testing normal closure membership, denoted by $\proc{NORM}(G)$, is as follows. 

\begin{codebox}
\Procname{Protocol $\proc{NORM}(G)$}
\zi \const{Alice's input:} a subgroup $H\in\Subgroups{G}$.
\zi \const{Bob's input:} an element $y\in G$.
\zi \const{Bob's output:} $z\in \{ 0,1\}$.
\li Alice chooses a random element $\mu$ of $\Lambda_{\normal{H}}$ with probability $[\mu(1)]^2\abs{\normal{H}}/\abs{G}$;
\li Alice sends the name of $\mu$ to Bob;
\li Bob outputs $1$ if $\mu(y)=\mu(1)$ and outputs $0$ otherwise.
	\End
\end{codebox}
\vspace{0mm}
Observe that by equation~\eqref{equation:char}, the weights of Step 1 do indeed determine a probability distribution.
Notice that $\abs{\Lambda_{\normal{H}}}\le\abs{G}$ since $\Lambda_{\normal{H}}\subseteq \IrrChar{G}$ and $\abs{\IrrChar{G}}\le \abs{G}$.
Thus Protocol $\proc{NORM}(G)$ can be implemented using $\ceil{\log_2\abs{G}}$ bits of communication.
We now show the correctness of this protocol.

\begin{proposition}
On any input $(H,y)$, Protocol $\proc{NORM}(G)$ outputs $1$ with probability $1$ if $y\in\normal{H}$, and outputs $0$
with probability at least $1/2$ if $y\notin\normal{H}$.
\end{proposition}
\begin{proof}
If $y\in \normal{H}$, then for any element $\mu$ in $\Lambda_{\normal{H}}$ the equality $\mu(y)=\mu(1)$ holds.
Then Bob always outputs 1. Protocol $\proc{NORM}(G)$ has thus perfect completeness.

Now suppose that $y\notin \normal{H}$.
Denote $B=\{\chi\in \Lambda_{\normal{H}}\st \chi(y)=\chi(1)\}$. The error probability of the protocol
is 
$$
\frac{\abs{\normal{H}}}{\abs{G}}\sum_{\chi\in B}[\chi(1)]^2.
$$
To conclude our proof,  we now prove that $$\sum_{\chi\in B}[\chi(1)]^2\le \frac{\abs{G}}{2\abs{\normal{H}}}.$$
Let $K$ denote the normal closure of the
set $\normal{H}\cup\{y\}$ in $G$.
Remember that the normal closure of a set $S\subseteq G$ is  the smallest normal subgroup of $G$ including $S$,
and can be defined explicitly as the subgroup of $G$ generated by all the elements $gzg^{-1}$ for $g\in G$ and $z\in S$.
Since $y\notin \normal{H}$ the subgroup $\normal{H}$ is a proper subgroup of $K$. In particular $\abs{K}/\abs{\normal{H}}\ge 2$.
We now claim that $B=\Lambda_{K}$.
Then Equation (\ref{equation:char}) implies that
$$
\sum_{\chi\in B}[\chi(1)]^2=
\sum_{\chi\in \Lambda_{K}}[\chi(1)]^2=[G:K]\le \frac{\abs{G}}{2\abs{\normal{H}}}.$$
The proof of the claim follows.
First suppose that $\chi$ is an element of $\Lambda_{K}$. Then $\chi(y)=\chi(1)$ and thus $\chi\in B$.
Now suppose that $\chi$ is an element of $B$. Then $\normal{H}\cup\{y\}\subseteq \ker(\chi)$. From the basic properties of characters, we conclude that
$K\subseteq \ker(\chi)$ and thus $\chi \in \Lambda_{K}$.
\end{proof}

Given a finite group $G$, let $\Subgroups{G}^\ast$ be the set of normal subgroups of $G$.
\ Since for a normal subgroup $H$ of $G$ we have $\normal{H}=H$, we conclude that Protocol $\proc{NORM}(G)$ solves the restriction of
$\MEMB_G$ to the domain $\Subgroups{G}^\ast\times G$ (notice that this is still a total function).
\begin{theorem}\label{theorem_normal}
For any finite group $G$,
the restriction of $\MEMB_G$
to the domain  $\Subgroups{G}^\ast\times G$ 
can be computed
with perfect completeness and soundness error $1/2$
by communicating at most $\ceil{\log_2\abs{G}}$ bits.
\end{theorem}

We now show a simple lower bound on the communication complexity of $\MEMB_G$.
We first recall the definition of the VC-dimension of a set of functions \cite{Vapnik+71}.
\begin{definition}
Let $\Sigma$ be a set of Boolean functions over a finite domain $Y$. We say that a set $S\subseteq Y$ is shattered
by $\Sigma$ if for every subset $R\subseteq S$ there exists a function $\sigma_R \subseteq \Sigma$ such that $\forall y \in S,  (\sigma_R(y) =1$  if and only if $y\in R)$.
The largest size of set $S$ over all $S$ shattered by $\Sigma$
is the VC-dimension of $\Sigma$,
denoted by $V C(\Sigma)$.
\end{definition}

We say that a subset $S$ of a finite group $G$ is an independent subset of $G$ if, for each $g\in S$, 
element $g$ cannot be written as any product of elements of $S\backslash \{g\}$.
We denote by $\gamma(G)$ the maximal size of an independent subset of $G$. 
Notice that, for any finite group $G$, the inequality $\gamma(G)\le \log_2\abs{G}$ holds. We now state our lower bound.
\begin{proposition}\label{proposition:lower}
$Q^1(\MEMB_G)=\Omega(\gamma(G))$.
In particular, the family of groups $G=\Int_2^r$ for $r\ge 1$ satisfies $Q^1(\MEMB_G)=\Omega(\log\abs{G})$.
\end{proposition}
\begin{proof}
For each subgroup $H\in \Subgroups{G}$, define the function $f_H:G\to\{0,1\}$ as $f_H(y)=\MEMB_G(H,y)$ for every $y\in G$.
Denote $\Sigma=\{f_H\:|\: H\in\Subgroups{G}\}$. 
A result by Klauck  \cite{KlauckSICOMP07} shows that $Q^1(\MEMB_G)\ge (1-h(1/3))\cdot VC(\Sigma)$, where $h$ is the binary entropy function.

Let $g_1,\ldots,g_{\gamma(G)}$ be distinct elements of $G$ such that $S=\{g_1,\ldots,g_{\gamma(G)}\}$
is a subset of independent elements of $G$.
The subset $S\subseteq G$ is shattered by $\Sigma$
since it is easy to show that, for any subset $R\subseteq S$, the function $f_{\gen{R}}$ is such that $\forall y \in S,  f_{\gen{R}}(y) =1$  if and only if $y\in R$ 
(here $\gen{R}$ denotes the subgroup generated by the elements in R).
Then $VC(\Sigma)\ge \gamma(G)$ and $Q^1(\MEMB_G)\ge (1-h(1/3))\cdot \gamma(G)$.

The second part of the proposition follows from the observation that each group $\Int_2^r$ is also a vector space of dimension $r$ over the finite field $\Int_2$ and, 
thus, $\gamma(\Int_2^r)=r=\log_2(\abs{\Int_2^r})$.
\end{proof}
Proposition~\ref{proposition:lower} shows that, for the family of groups $G=\Int_2^r$, Protocol $\proc{NORM}(G)$
is optimal up to a constant factor.

\section{Algorithms for groups with small modular representations}\label{section:modular}
In this section we present a protocol computing the group membership function for groups with small modular representations. \ 
Let $\F_q$ be a finite field with  characteristic $p$ not dividing $|G|$.
Our protocol, denoted by $\proc{MOD-REP}(G,\F_q)$, is the following.

\begin{codebox}
\Procname{Protocol $\proc{MOD-REP}(G,\F_q)$}
\zi \const{Alice's input:} a subgroup $H\in\Subgroups{G}$.
\zi \const{Bob's input:} an element $y\in G$.
\zi \const{Bob's output:} $z\in \{0,1 \}.$
\li 
Alice chooses a representation $\rho:G\to\GL(V_\rho)$ in $\Irr{G}{\F_q}$ with probability
$
\frac{|H| \cdot d_\rho \cdot \dim I_H(\rho)}{|G|}\,;
$
\li
Alice chooses a random vector $\vec{v} \in I_H(\rho)\subseteq V_\rho$;
\li 
Alice sends the name of $\rho$ and the vector $\vec{v}$ to Bob;
\li Bob outputs 1 if $\rho(y)\vec{v}=\vec{v}$ and outputs 0 otherwise.
\End
\end{codebox}\vspace{0mm}
Observe that by equation~\eqref{eq:stable}, the weights of Step 1 do indeed determine a probability distribution.

We now show the correctness of this protocol.
\begin{theorem}
Let $G$ be a finite group and $\F_q$ be a finite field of characteristic $p$ not dividing $|G|$.
Protocol $\proc{MOD-REP}(G,\F_q)$ computes $\MEMB_G$ with perfect completeness and constant soundness error.
Its communication complexity is at most $\ceil{\log_2\abs{G}}+d^q_{\max}\cdot \ceil{\log_2 q}$ bits,
where $d^q_{\max}$ is the maximum dimension of an irreducible $\F_q$-representation of $G$.
\end{theorem}
\begin{proof}
Note that the protocol is clearly complete: if $y \in H$, then Bob always accepts.

To establish soundness, let $y\notin H$ and define $K = \langle H, y\rangle$, the smallest subgroup containing both $H$ and $y$.
Remember that $I_K(\rho)$ denotes the subspace of $V_\rho$ pointwise fixed by $K$.
We see that
$$
\Exp_\rho \left [  \frac{\dim I_K(\rho)}{\dim I_H(\rho)} \right] = \sum_\rho \frac{|H| d_\rho \dim I_K(\rho)}{|G|} = \frac{|H|}{|K|} =  \frac{1}{[K:H]} \leq \frac{1}{2}\, ,
$$
again by equation~\eqref{eq:stable}. Observe, then, that $I_K(\rho) \subseteq I_H(\rho)$ and so
$$
 \Exp_\rho \left [  \frac{\dim I_K(\rho)}{\dim I_H(\rho)}\right]\ge \Pr[I_K(\rho) = I_H(\rho)].
$$
Then
$$
\Pr[I_K(\rho) \neq I_H(\rho)] = 1-\Pr[I_K(\rho) = I_H(\rho)]\ge 1/2\,.
$$
When $I_K(\rho) \neq I_H(\rho)$, the vector $\vec{v}$ chosen by Alice has probability no more than $1/q$ to be in $I_K(\rho)$.
Then $\rho(y) \vec{v} \neq \vec{v}$ with constant probability in her choices of $\rho$ and $\vec{v}$. 

Since $\abs{\Irr{G}{\F_q}}\le \abs{G}$, the communication complexity of the protocol is  
at most $\ceil{\log_2\abs{G}}+d^q_{\max}\cdot \ceil{\log_2{q}}$.
\end{proof}



In light of the complexity guarantee of the protocol above, it is natural to ask how the dimensions of the irreducible representations of a finite group $G$ compare over various fields and, especially, how the modular case compares to the 
complex case.
When the group algebras involved are semi-simple (as they are in this paper due to our insistence that $p \not|\:\: |G|$), there is a tight connection expressed in the following proposition. 


\begin{proposition}
Let $G$ be a finite group of exponent $m$ and $p$ be any prime not dividing $|G|$. Then the relation $d_{\max}^p \leq d_{\max}^0 \ord_m(p)$ holds, 
where $d^0_{\max}$ is the maximum dimension of a complex irreducible representation of $G$,  $d^p_{\max}$ is the maximum dimension of an irreducible $\F_p$-representation of $G$,
and $\ord_m(p)$ is the order of $p$ in $\Z_m^*$, the multiplicative group of the integers relatively prime to $m$.
\end{proposition}

\begin{proof}
This is a consequence of the ``$c$-$d$-$e$ triangle'' (see~\cite{Serre77}). See Appendix~\ref{sec:cde} for a brief discussion.
\end{proof}

As there always exists a prime $p$ of size $O(\log|G|)$ that does not divide $|G|$, we obtain the following corollary.
\begin{corollary}
$R^1(\MEMB_G)=O(d_{\max}^0 \cdot m \cdot \log\log |G|)$, where $m$ denotes the exponent of $G$ and $d^0_{\max}$ is the maximum dimension of a complex irreducible representation of $G$.
\end{corollary}

\section{Algorithms for groups with small $\C$-representations}
We now focus on the case where the dimensions of the irreducible $\Comp$-representations of $G$ is under control.
The key idea is to discretize the protocol given in the previous section. To achieve this goal we use the concept of an
$\epsilon$-net of a sphere. (As our nets will lie in the vector spaces acted upon by the irreps of $G$, we define them as subsets of complex Hilbert spaces.)
\begin{definition}
Let $V$ be a finite-dimensional complex Hilbert space.
An \emph{$\epsilon$-net} of $V$ is a finite family of unit-vectors $N\subseteq V$ so that for every unit-length vector $\vec{w} \in V$, there is a vector $\vec{n} \in N$ so that 
$|\langle \vec{n}, \vec{w}\rangle|^2 >  1- \epsilon^2$.
\end{definition}

\begin{proposition}\label{prop:net}
For any $\epsilon>0$ and for any complex Hilbert space $V$ of dimension $d$, there exists an $\epsilon$-net of size at most $(4/\epsilon)^{2d}$.
\end{proposition}

\begin{proof}
For any dimension $d$ and distance $\epsilon > 0$, there is a set of points $A \subset S^{d-1}$ of cardinality no more than $(4/\epsilon)^d$ with the property that every point of $S^{d-1}$ has distance no more than $\epsilon$ from some point of $A$ (see, e.g., \cite[\S3.1]{Matousek02}). This yields a set with analogous properties of size no more than $(4/\delta)^{2d-1}$ for the complex $d$-sphere, which has the same metric as the real $2d-1$ sphere. Note that if $\vec{v}$ and $\vec{w}$ are two unit vectors of $V$, we may write $\vec{v} = \langle\vec{v},\vec{w}\rangle \vec{w} + \vec{r}$ with $\langle \vec{r},\vec{w}\rangle = 0$ in which case, $\|\vec{r}\| \leq  \| \vec{v} - \vec{w}\|$. The statement of the proposition follows. 
\end{proof}

Our protocol requires the choice of a sufficiently dense $\epsilon$-net for each irreducible representation in $\Irr{G}{\Comp}$.
This choice is independent of the inputs of the protocol and so can be done by Alice and Bob without communication.  
The protocol is as follows. 

\begin{codebox}
\Procname{Protocol $\proc{COMP-REP}(G,\epsilon)$}
\zi \const{Alice's input:} a subgroup $H\in\Subgroups{G}$
\zi \const{Bob's input:} an element $y\in G$
\zi \const{Bob's output:} $z\in \{0,1 \}.$
\li Alice and Bob agree on an $\epsilon$-net $N_\rho$ of $V_\rho$ for each $\rho:G\to\GL(V_\rho)$ in $\Irr{G}{\Comp}$;
\li 
Alice chooses a representation $\rho\colon G\to\GL(V_\rho)$ in $\Irr{G}{\Comp}$ with probability
$
\frac{|H| \cdot d_\rho \cdot \dim I_H(\rho)}{|G|}\,;
$
\li
Alice chooses a random (according to Haar measure) unit length vector $\vec{v} \in I_H(\rho)\subseteq V_\rho$;
\li
Alice sends Bob the name of $\rho$ and the closest vector $\vec{n}$ in $N_\rho$ to the vector $\vec{v}$;
\li 
If $|1 - \langle \rho(y) (\vec{n}), \vec{n}\rangle| \leq 2\epsilon$, then Bob outputs 1;
\zi
Otherwise $|1 - \langle \rho(y)(\vec{n}),\vec{n}\rangle| > 2\epsilon$, and Bob outputs 0. 
\End
\end{codebox}\vspace{0mm}
Observe that by equation~\eqref{eq:stable}, the weights at Step 2 do indeed determine a probability distribution on $\Irr{G}{\Comp}$. \
Ideally, at Step 3,  Alice would communicate $\vec{v}$ to Bob: Bob could then check if $\rho(y) (\vec{v}) = \vec{v}$ and, if so, would figure that $y \in H$.  
If $\rho(y)(\vec{v}) \neq \vec{v}$, Bob would be sure that $y \not \in H$, since $I_H(\rho)$ is precisely the fixed space of $H$. The proof below shows that by sending a sufficiently close approximation to $\vec{v}$, Bob can still answer confidently.

The following theorem states the correctness and the communication complexity of this protocol.

\begin{theorem}\label{th:complex}
There exists a choice of $\epsilon_G$ such that 
Protocol $\proc{COMP-REP}(G,\epsilon_G)$ computes $\MEMB_G$ with perfect completeness and constant soundness error by communicating $O(d^0_{\max}\cdot \log \abs{G})$ bits,
where $d^0_{\max}$ denotes the maximum dimension of an irreducible $\Comp$-representation of $G$.
\end{theorem}
\begin{proof}
As the name of the representation $\rho$ can be encoded using $\ceil{\log_2{\abs{G}}}$ bits,
the communication complexity of the protocol will be dominated by the number of bits necessary to encode the vector $\vec{n}$.
We will show that a choice $\epsilon=\epsilon_G=\Omega(1/(|G|^2 \poly \log |G|))$ suffices to achieve perfect completeness and constant soundness.
According to Proposition~\ref{prop:net}, such an $\epsilon$-net 
can be indexed with $O(d_\rho \log |G|)$ bits.
This gives our upper bound.

We proceed with the analysis of the completeness and soundness of the protocol.

\medskip
\noindent{\bf Completeness}
Observe that if $y \in H$, then the vector $\vec{v}$ chosen by Alice in the protocol is fixed by $\rho(y)$. Recall that Alice sends Bob a vector $\vec{n}$ for which $|\langle \vec{n}, \vec{v}\rangle|^2 \geq 1 - \epsilon^2$; writing
$$
\vec{n} = \langle \vec{n}, \vec{v}\rangle \vec{v} + \vec{r}
$$
(where $\langle \vec{r}, \vec{v} \rangle = 0$) we have
$$
1 = \langle \vec{n}, \vec{n} \rangle = |\langle \vec{n}, \vec{v}\rangle|^2 + \langle\vec{r}, \vec{r}\rangle
$$
and $\| \vec{r} \| \leq \epsilon$. Considering that
$$
\langle \rho(y) \vec{n}, \vec{n} \rangle = \left| \langle \vec{n}, \vec{v}\rangle \right|^2 + \langle \rho(y) \vec{r}, \vec{n}\rangle
$$
we conclude that
\[
\left| 1 - \langle \rho(y)\vec{n}, \vec{n}\rangle \right| =  \left| 1 - | \langle \vec{n}, \vec{v}\rangle|^2  - \langle \rho(y) \vec{r}, \vec{n}\rangle\right|  \leq  \left( 1 - | \langle \vec{n}, \vec{v}\rangle|^2\right) + \left|\langle \rho(y) \vec{r}, \vec{n}\rangle\right| \leq \epsilon^2 + \left|\langle \rho(y) \vec{r}, \vec{n}\rangle\right|\,.
\]
Recall that $\rho(y)$ is unitary, so that $\| \rho(y) \vec{r}\| = \| \vec{r}\|$.
Then, by the Cauchy-Schwarz inequality,
$$
\left| 1 - \langle \rho(y)\vec{n}, \vec{n}\rangle \right| 
\leq \epsilon^2 + \| \vec{r} \| \leq \epsilon^2 + \epsilon\,.
$$
As $\epsilon < 1$, we have $\epsilon^2 + \epsilon \leq 2\epsilon$ and it follows that the protocol has perfect completeness.

\medskip
\noindent{\bf Soundness} We wish to  show that for sufficiently small $\epsilon$ $(= 1/\poly{|G|})$, the protocol has constant soundness. Assume that $y \not \in H$ and let $K = \langle H, y\rangle$, the smallest subgroup containing $H$ and $y$. Our goal will be to show that with constant probability $\langle \vec{v}, \rho(y) \vec{v} \rangle$ is far from $1$, in which case the same can be said of $\vec{n}$ so long as $\epsilon$ is sufficiently small.

From equation~\eqref{eq:stable},
$$
\Exp_\rho \left [  \frac{\dim I_K(\rho)}{\dim I_H(\rho)} \right] = \sum_\rho \frac{|H| d_\rho \dim I_K(\rho)}{|G|} = \frac{|H|}{|K|} =  \frac{1}{[K:H]} \leq \frac{1}{2}\, .
$$
\begin{sloppypar}
Then, with constant probability, the subspace of $I_H(\rho)$ fixed by $y$ has dimension no more than $2/3 \cdot \dim I_H(\rho)$.
We may write the vector $\vec{v}\in I_H(\rho)$ as $\vec{v} = \vec{v}_y + \vec{v}'$, where $\vec{v}_y \in I_K(\rho)$ and $\vec{v}' \in[ I_K(\rho)]^\perp$, the space perpendicular to $I_K(\rho)$. 
We then have $\rho(y) \vec{v}_y = \vec{v}_y$ and $\vec{v}_y \in I_K(\rho) \subset I_H(\rho)$. Now, as $\vec{v}$ is chosen uniformly on the unit sphere in $V_\rho$, we have $\Exp_\vec{v}[\| \vec{v}_y\|^2] =  \dim I_K(\rho) / \dim I_H(\rho)$ and the probability
\end{sloppypar}
$
\Pr_{\rho,\vec{v}}[\|\vec{v}'\|^2 \ge 1/6]
$
is lower bounded by a constant.\footnote{Of course, when $\dim I_K(V_\rho) < 2/3 \dim I_H(V_\rho)$, the random variable $\|v'\|^2$ possesses much stronger concentration around the expected value than this.}
We wish to conclude that, conditioned on the event $\|\vec{v}'\|^2 \geq 1/6$, the value
$$
\frac{\langle \vec{v}', \rho(y) \vec{v}' \rangle}{\| \vec{v}'\|^2}
$$
cannot be too close to 1. We will show, in fact, that the real part is appropriately bounded below 1.

Consider the restriction of the representation $\rho: G\to \GL(V_\rho)$ chosen by Alice to the subgroup $K$: specifically, we may decompose $V_\rho$ as an orthogonal direct sum of $K$-invariant subspaces:
\begin{equation*}
V_\rho = \bigoplus_{i} W_{\sigma_i}\,,
\end{equation*}
where each $\sigma_i$ is in $\Irr{K}{\Comp}$ (but copies of the same irrep may appear several times in the direct sum). 
In this decomposition, $\vec{v}_y$ is precisely the projection of $\vec{v}$ into the subspace $\bigoplus_{i\colon \sigma_i=1} W_{\sigma_i}$ corresponding to the copies of the trivial representation; 
$\vec{v}'$, on the other hand, lies solely in $\bigoplus_{i\colon\sigma_i\neq 1} W_{\sigma_i}$. 
As both $\vec{v}$ and $\vec{v}_y$ lie in $I_H(\rho)$, the difference $\vec{v}'$ does as well and the projection of $\vec{v}'$ into each $W_{\sigma_i}$ is $H$-invariant (that is, lies in $I_H(\sigma_i)$). With this in mind, we shall upper bound
$$
\frac{\Re \langle \vec{v}', \rho(y) \vec{v}'\rangle }{\| \vec{v}' \|^2}
$$
by controlling
\begin{equation*}
\lambda_y \triangleq \max_{\sigma \neq 1} \max_{\vec{w}\in I_H(\sigma)} \frac{\Re \langle \vec{w}, \sigma(y) \vec{w}\rangle}{\|\vec{w}\|^2}
\end{equation*}
taken over all nontrivial irreps $\sigma$ of $K$ and all $H$-invariant vectors $\vec{w}$ in $W_\sigma$. 
In particular, writing $\vec{v}' = \sum_{i\colon \sigma_i\neq 1} \vec{v}'_i$ 
(with each $\vec{v}'_i$ lying in $W_{\sigma_i}$), we have $\| \vec{v}' \|^2 = \sum_{i\colon \sigma_i\neq 1} \| \vec{v}'_i\|^2$ and
$$
\Re \langle \vec{v}', \rho(y) \vec{v}'\rangle = \sum_{i\colon \sigma_i\neq 1} \Re \langle \vec{v}'_i, \rho(y) \vec{v}'_i \rangle \leq 
\sum_{i\colon \sigma_i\neq 1} \lambda_y \| \vec{v}'_i \|^2 = \lambda_y \| \vec{v}' \|^2\,.
$$

Observe that if $A$ is a set of generators for $H$ and $\vec{w}$ is an $H$-invariant vector of $W_\sigma$,
$$
\langle \vec{w}, \sigma(y) \vec{w}\rangle = \langle \vec{w}, \sigma(y) S_A \vec{w}\rangle
$$
where $S_A = S^\sigma_A= \frac{1}{|A|} \sum_{a \in A} \sigma({a})$. Then
$$
\lambda_y \leq \max_{\sigma \neq 1} \max_{\vec{w}\in W_\sigma} \frac{\Re \langle \vec{w}, \sigma(y) S_A \vec{w}\rangle}{\|\vec{w}\|^2}\,.
$$
(Note that the vector $\vec{w}$ is not constrained to be $H$-invariant in this expression.)
If we choose $A$ to be a symmetric generating set (so that $a \in A \Leftrightarrow a^{-1} \in A$) then
$S_A$ is self-adjoint and $\sigma(y)$ is unitary so that 
\begin{equation*}
\max_{\sigma \neq 1} \max_{\|\vec{w}\|=1} \Re \langle \vec{w}, \sigma(y) S_A \vec{w}\rangle = 
\max_{\sigma \neq 1} \max_{\|\vec{w}\|=1} \frac{1}{2} \Bigl[\langle \vec{w}, \sigma(y) S_A \vec{w}\rangle + \langle \vec{w}, S_A\sigma(y^{-1})  \vec{w}\rangle\Bigr].
\end{equation*}
As the operator $\sigma(y) S_A + S_A \sigma(y^{-1})$ is Hermitian, we have
\begin{equation*}
\max_{\sigma \neq 1} \max_{\|\vec{w}\|=1} \Re \langle \vec{w}, \sigma(y) S_A \vec{w}\rangle
\le \max_{\sigma \neq 1}  \left\| \frac{\sigma(y) S_A + S_A \sigma(y^{-1})}{2}\right\|
\end{equation*}
where $\| \cdot \|$ denotes the operator norm. 

In order to control this operator norm, observe that the linear operator $(1/2)\left[ \sigma(y) S_A + S_A \sigma(y^{-1})\right]$ is precisely given by the left action of the group algebra element
\begin{equation}\label{eq:walk}
\mathbf{[A,y]} \triangleq \frac{1}{2 |A|} \left[\sum_{a \in A} e_{ya}+ \sum_{a \in A} e_{ay^{-1}} \right] \in \C[K]
\end{equation}
on the invariant subspace $W_\sigma$ of $\C[K]$ corresponding to the representation $\sigma$. Alternatively, we may consider the Cayley graph on the group $K$ given by the symmetric generating (multi-)set $yA \cup Ay^{-1}$. The (normalized) adjacency matrix of this Cayley graph is identical to the regular representation evaluated at the group algebra element~\eqref{eq:walk} above. As $yA \cup Ay^{-1}$ is a (symmetric) generating set for $K$, the operator norm of $\sigma(\mathbf{[A,y]})$ is bounded below 1 for each nontrivial $\sigma$ (see, e.g.,~\cite{Hoory:2006yq}). In order to conclude the proof, we require explicit bounds on this spectral gap.

A result of Erd\H{o}s and Renyi~\cite{Erdos:1965kx} asserts that we may select a set of generators $A$ for $H$ of size $O(\log |H|)$ so that the diameter of the resulting Cayley graph (generated by $A$ over $H$) is $O(\log |H|)$. Considering that the diameter of $A$ (as generators for $H$) is $O(\log |H|)$, it is easy to see that the set $yA \cup Ay^{-1}$ induces a Cayley graph on $K$ of diameter no more than $O([K : H] \log |H|)$.

Now we may invoke a theorem of Babai~\cite{Babai:1991kx} asserting that the second eigenvalue of any (undirected) Cayley graph with degree $d$ and diameter $\Delta$ is no more than $d - \Omega(1/\Delta^2)$. (If we normalize the adjacency matrix by degree, the second eigenvalue is no more than $1 - \Omega(1/(d \Delta^2))$.) We conclude that
\begin{equation*}
\Re \frac{\langle \vec{v}', \rho(y)\vec{v}'\rangle}{\| \vec{v}' \|^2} \leq  \lambda_y \leq \max_{\sigma \neq 1}  \left\| \frac{\sigma(y) S_A + S_A \sigma(y^{-1})}{2}\right\| \leq 1 - \Omega\left(\frac{1}{[K:H]^2 \log^3 |H|}\right)
\end{equation*}
and, considering that $\|\vec{v}'\|^2 \geq 1/6$, that
$$
\Re \langle \vec{v}, \rho(y)\vec{v} \rangle \le
\|\vec{v_y}\|^2+\frac{\Re \langle \vec{v}', \rho(y)\vec{v}'\rangle }{\| \vec{v}' \|^2}\cdot \| \vec{v}' \|^2
\leq 1 - \Omega\left(\frac{1}{[K:H]^2 \log^3 |H|}\right)\,.
$$
Finally, Alice's $\vec{n}$ can be written $\vec{n} = \vec{v} +\vec{r}$ with $\|\vec{r}\| \leq \epsilon$, in which case
$$
|\langle \vec{n}, \rho(y)\vec{n} \rangle| \leq 1 - \Omega\left(\frac{1}{[K:H]^2 \log^3 |H|}\right) + 3\epsilon \leq 1 - 2\epsilon\,,
$$
for $\epsilon^{-1} = \Omega([K:H]^2 \log^3 |H|)$; thus the protocol is sound.
\end{proof}
In particular, Theorem \ref{th:complex} shows that, over groups for which $d^0_{\max}$ is constant, the subgroup membership problem can be solved 
using $O(\log\abs{G})$-bit communication.
There is a very beautiful characterization of such groups: a family of groups has representations of bounded degree 
if and only each group of the family has an abelian subgroups of constant index~\cite{Gluck:1985la}. 
We thus obtain the following corollary.
\begin{corollary}
Let $G$ be a family of groups each possessing an abelian subgroup of constant index. Then $R^1(\MEMB_G)=O(\log\abs{G}).$
\end{corollary}

\section*{Acknowledgments}
The authors are grateful to Keith Conrad, Iordanis Kerenidis, and Troy Lee for helpful discussions on this subject.

\newcommand{\etalchar}[1]{$^{#1}$}

\appendix
\section{Remarks on the relationship between $\C$ and $\F_p$ representations}
\label{sec:cde}

Let $G$ be a finite group of exponent $m$ (so $m$ is the smallest integer for which $g^m = 1$ for all $g \in G$). We outline a technique for reducing $\C$-representations of $G$ to $\F_p$-representations in a manner that preserves irreducibility. For a complete account, see~\cite{Serre77}. By a difficult theorem of Brauer (see, e.g., \cite{CR06}), one may always realize a $\C$-irrep over the field $\Q[\zeta_m]$ where $\zeta_m$ is a principal $m$th root of unity. (It is natural to guess that this might be true, as all eigenvalues of a representation of $G$ are $m$th roots of unity.) Let $\Z[\zeta_m]$ be the ring of algebraic integers in $\Q[\zeta_m]$ (it so happens that in this cyclotomic case $\Z[\zeta_m]$ is indeed the ring of algebraic integers). Let $p > 2$ be a prime, and let $\mathfrak{P} = \Z[\zeta_m](p))$; this is a prime ideal of $\Z[\zeta_m]$ lying over $p$ in the sense that $\mathfrak{P} \cap \Z = (p)$. Now, if only the representation could be realized over $\Z[\zeta_m]$, we could simply reduce mod $\mathfrak{P}$ and obtain a representation over an extension of $\F_p$. However, this is either not always true or just not known to be true by the authors. To fix the problem, one first localizes at $\mathfrak{P}$; that is, we consider the ring $\Z[\zeta_m]_\mathfrak{P}$ of all fractions with the property that the denominator lies outside $\mathfrak{P}$; this is a principal ideal domain with a single prime (and maximal) ideal $\mathfrak{P}$. In this case, the representation can be realized over $\Z[\zeta_m]_\mathfrak{P}$, as this PID generates the whole field as its field of fractions (see~\cite[\S73.6]{CR06}). Now we can reduce mod $\mathfrak{P}$; the result is a matrix realization over the field $\Z[\zeta_m]_\mathfrak{P}/\mathfrak{P}$; it is easy to check that this is an extension of the field $\F_p = \Z/(p)$. Furthermore, the dimension of this extension field is the multiplicative order of $p$ modulo $m$ (the same as the extension of the splitting field of the polynomial $X^{m} -1$ over $\F_p$). This immediately gives rise to a representation over the field $\F_q$ with $q = p^{\ord_m p} \leq p^{\phi(m)}$. We remark that this process preserves irreducibility, and induces a complete decomposition of $\F_p[G]$ into irreducible representations.

\end{document}